\documentclass[11pt,envcountsect,envcountsame]{llncs}
\usepackage[hmargin=1in,vmargin=1.3in]{geometry}

\usepackage{latexsym}
\usepackage{amssymb}
\usepackage{amsmath}
\usepackage{float}
\usepackage{framed}
\usepackage{qip}
\usepackage{bbm}
\usepackage{tabularx}
\usepackage{xspace}
\usepackage[sectionbib]{natbib}
\usepackage{hyperref}

\newcounter{itm}
\newenvironment{myprotocol}[1]
  {\begin{minipage}{\columnwidth} 
    \begin{framed}\hspace{0ex} 
     \begin{minipage}{0.99\columnwidth}
       {\bf #1:}
       \setcounter{itm}{1}
       \begin{list}{\arabic{itm}.}{\usecounter{itm}
          \setlength{\itemsep}{0mm}
          \setlength{\leftmargin}{\labelwidth}
          \setlength{\topsep}{\parsep}}}           
   {    \end{list}
       \vspace{-1.5ex} 
       \end{minipage} 
     \end{framed} 
    \end{minipage}\vspace{-0.6ex}}

\newenvironment{myfigure}[1]    
         {\begin{figure}[#1] \centering}
         { \end{figure}}

\renewcommand{\S}{\ensuremath{{\sf S}}}
\newcommand{\U}{\ensuremath{{\sf U}}}
\newcommand{\dS}{\ensuremath{{\sf S}^*}}
\newcommand{\dU}{\ensuremath{{\sf U}^*}}

\newcommand{\QID}{{\sf \textit{Q-ID}}\xspace}
\newcommand{\xQID}{{\sf \textit{Q-ID}$^+$}\xspace}

\newcommand{\epsclose}[1][\varepsilon]{\approx_{#1}}

\newcommand{\close}[1]{\approx_{#1}}

\newcommand{\negl}{negl}

\newcommand{\set}[1]{\{#1\}}
\newcommand{\Set}[2]{\{ #1 : #2\}}

\newcommand{\eps}{\varepsilon}

\newcommand{\E}{\mathrm{E}}
\newcommand{\regE}{E}

\newcommand{\regS}{\ensuremath{E}_{\dS}}
\newcommand{\regU}{\ensuremath{E}_{\dU}}

\newcommand{\code}{\mathfrak{c}}

\newcommand{\W}{{\cal W}}
\newcommand{\mX}{\mathcal{X}}
\newcommand{\mY}{\mathcal{Y}}

\newcommand{\ev}{{\cal E}}

\newcommand{\Hmin}[1][]{H_{\min}^{#1}}
\newcommand{\guess}{p_\mathrm{guess}}

\newcommand*{\assign}{\ensuremath{\kern.5ex\raisebox{.1ex}{\mbox{\rm:}}\kern -.3em =}}

\newcommand{\delete}[1]{}
\newcommand{\remove}[1]{}

\newcommand{\myparagraph}[1]{\medskip\noindent{\sc #1}}

\def\version$#1,v #2 #3/#4/#5 #6${#2 (#5-#4-#3)}

\title{
Secure Identification and QKD in the \\ Bounded-Quantum-Storage
Model\thanks{A preliminary version of this paper appeared in
  \emph{Advances in Cryptology---CRYPTO 2007}~\cite{DFSS07}.}
}

\author{Ivan Damg{\aa}rd\inst{1} \and Serge Fehr\inst{2}\fnmsep
  \thanks{Supported by the Dutch Organization for Scientific Research
    (NWO).}  \and Louis Salvail\inst{3}\and Christian
  Schaffner\inst{2}\fnmsep\thanks{supported by EU fifth framework
    project QAP IST 015848 and the NWO VICI project 2004-2009.} }

\institute{
  DAIMI, Aarhus University, Denmark\\
  \email{ivan@cs.au.dk} \and
  Centrum Wiskunde \& Informatica (CWI) Amsterdam, The Netherlands\\
  \email{\{s.fehr|c.schaffner\}@cwi.nl} \and
  Universit\'e de Montr\'eal (DIRO), QC, Canada\\
  \email{salvail@iro.umontreal.ca} }

\authorrunning{I.B.~Damg{\aa}rd, S.~Fehr, L.~Salvail,  and C.~Schaffner}
\tocauthor{Ivan B.~Damg{\aa}rd (Aarhus University, Denmark),
Serge Fehr (CWI Amsterdam, The Netherlands),
Louis Salvail (Universit\'e de Montr\'eal, QC, Canada),
Christian Schaffner (CWI Amsterdam, The Netherlands)}

\begin{document}

\maketitle

\setcounter{footnote}{0}

\begin{abstract}
  We consider the problem of secure identification: user $\U$ proves
  to server $\S$ that he knows an agreed (possibly low-entropy) password $w$, while
  giving away as little information on
  $w$ as possible, namely the adversary can exclude at most one possible password
  for each execution of the scheme. We
  propose a solution in the bounded-quantum-storage model, where $\U$
  and $\S$ may exchange qubits, and a dishonest party is assumed to
  have limited quantum memory. No other restriction is posed upon the
  adversary.
  An improved version of the proposed identification scheme is also
  secure against a man-in-the-middle attack, but requires $\U$ and
  $\S$ to additionally share a high-entropy key $k$. However, 
  security is still guaranteed if one party loses
  $k$ to the attacker but notices the loss. In both versions of the
  scheme, the honest participants need no quantum memory, and
  noise and imperfect quantum sources can be tolerated.
  The schemes compose sequentially, and $w$
  and $k$ can securely be re-used.
A small modification to the identification scheme results in a
quantum-key-distribution (QKD) scheme, secure in the
bounded-quantum-storage model, with the same re-usability properties
of the keys, and without assuming authenticated channels. This is in
sharp contrast to known QKD schemes (with unbounded adversary) without
authenticated channels, where authentication keys must be updated, and
unsuccessful executions can cause the parties to run out of keys. 
\end{abstract}

%%%%%%%%%%%%%%%%%%%%%%%%%%%%%%%%%%%%%%%%%%%%%%%%%%%%%%%%%%%%%%%%%
\section{Introduction}
%%%%%%%%%%%%%%%%%%%%%%%%%%%%%%%%%%%%%%%%%%%%%%%%%%%%%%%%%%%%%%%%%

\myparagraph{Secure Identification. } Consider two parties, a {\em
  user} $\U$ and a {\em server} $\S$, who share a common secret-key (or
password or Personal Identification Number PIN) $w$. In order to
obtain some service from $\S$, $\U$ needs to convince $\S$ that he is
the legitimate user $\U$ by ``proving'' that he knows $w$. In
practice---think of how you prove to the ATM that you know your
PIN---such a proof is often done simply by announcing $w$ to $\S$.
This indeed guarantees that a dishonest user $\dU$ who does not know
$w$ cannot identify himself as $\U$, but of course incurs the risk
that $\U$ might reveal $w$ to a malicious server $\dS$ who may now
impersonate $\U$. Thus, from a secure identification scheme we also
require that a dishonest server $\dS$ obtains (essentially) no
information on $w$.

There exist various approaches to obtain secure identification
schemes, depending on the setting and the exact security
requirements. For instance zero-knowledge proofs (and some weaker
versions), as initiated by Feige, Fiat and Shamir~\cite{FS86,FFS87}, allow for secure identification. In a more sophisticated model, where we allow the common key $w$ to be of low entropy and additionally consider a man-in-the-middle attack, we can use techniques from password-based key-agreement (like~\cite{KOY01,GL03}) to obtain secure identification schemes.
Common to these approaches is that security relies on the assumption that some computational problem (like factoring or computing discrete logs) is hard and that the attacker has limited computing power.

\myparagraph{Our Contribution. }
In this work, we take a
new approach: we consider quantum communication,
and we
develop two identification schemes which are information-theoretically
secure under the {\em sole} assumption that the attacker can only reliably store
quantum states of limited size. This model was first considered in~\cite{DFSS05}.
On the other hand, the honest participants only need to send qubits
and measure them immediately upon arrival, no quantum storage or
quantum computation is required. Furthermore, our identification
schemes are robust to both noisy quantum channels and
imperfect quantum sources. Our schemes can therefore be implemented
in practice using existing technology.

The first scheme is secure against dishonest users and servers but not against 
a man-in-the-middle attack. It allows the common secret-key $w$ to be non-uniform 
and of low entropy, like a human-memorizable password.
Only a user knowing $w$ can succeed in convincing the server. In any 
execution of this scheme, a dishonest user or server cannot 
learn more on $w$ than excluding one possibility, which is unavoidable. 
This is sometimes referred to as {\em password-based} identification. 
The second scheme requires in addition to $w$ a uniformly distributed high-entropy common secret-key $k$, but is additionally secure against a man-in-the-middle attack. Furthermore, security against a dishonest user or server holds as for the first scheme even if the dishonest party knows $k$ (but not~$w$). This implies that $k$ can for instance be stored on a smart card, and security of the scheme is still guaranteed even if the smart card gets stolen, assuming that the affected party notices the theft and thus does not engage in the scheme anymore. Both schemes compose sequentially, and $w$ (and $k$) may be safely re-used super-polynomially many times, even if the identification fails (due to an attack, or due to a technical failure). 

A small modification of the second identification scheme results in a
quantum-key-distribution (QKD) scheme
secure against bounded-quantum-memory adversaries. The advantage of the proposed
new QKD scheme is that no authenticated channel is needed and the
attacker can {\em not} force the parties to run out of
authentication keys.  The honest parties merely need to share a
password $w$ and a high-entropy secret-key $k$, which they can safely
re-use (super-polynomially many times), independent of whether QKD
succeeds or fails. Furthermore, like for the identification scheme,
losing $k$ does not compromise security as long as the loss is
noticed by the corresponding party.  One may think of this as a quantum
version of password-based authenticated key exchange. The properties of
our solution are in sharp contrast to all
known QKD schemes without authenticated channels (which do not
pose any restrictions on the attacker). In these schemes, an
attacker can force parties to run out of authentication keys by making
the QKD execution fail (e.g.~by blocking some messages). Worse,
even if the QKD execution fails only due to technical problems, the parties
can still run out of authentication keys after a short while, since they
cannot exclude that an eavesdropper was in fact present. This problem
is an important drawback of QKD implementations, especially of those
susceptible to single (or few) point(s) of failure~\cite{EPT03}.

\myparagraph{Other Approaches. }
We briefly discuss how our identification schemes compare with other approaches. We have already given some indication on how to construct {\em computationally} secure identification schemes. This approach typically allows for very practical schemes, but requires some unproven complexity assumption. Another interesting difference between the two approaches: whereas for (known) computationally-secure password-based identification schemes the underlying computational hardness assumption needs to hold indefinitely, the restriction on the attacker's quantum memory in our approach only needs to hold {\em during} the execution of the identification scheme, actually only at one single point during the execution. In other words, having a super-quantum-storage-device at home in the basement only helps you cheat at the ATM if you can communicate with it on-line quantumly -- in contrast to a computational solution, where an off-line super-computer in the basement can make a crucial difference.

Furthermore, obtaining a satisfactory identification scheme requires {\em some} restriction on the adversary, even in the quantum setting: considering only passive attacks, Lo~\cite{Lo97} showed that for an unrestricted
adversary, no  password-based quantum identification scheme exists.
In fact, Lo's impossibility result only applies  
if the user $\U$ is guaranteed not to learn anything about the outcome of the identification procedure. 
We can argue, however, that a different impossibility result holds even without Lo's restriction: 
We first show that secure computation of a classical
{\sc and} gate (in which both players learn the output) can be reduced to a password-based identification scheme. The reduction works as follows. 
Let $w_{0}$, $w_{0}'$ and $w_1$ be three distinct elements from ${\cal W}$. If Alice has private input $x_A=0$ then she sets $w_A=w_0$ and if $x_A=1$ then she sets $w_A = w_1$, and if Bob has private input $x_B=0$ then he sets $w_B=w_0'$ and if $x_B=1$ then he sets $w_B = w_1$. 
Then, Alice and Bob run the identification scheme on inputs $w_A$ and $w_B$, and if the identification is rejected, the output is set to $0$ while
if it is accepted, the output is set to $1$. Security of the identification scheme is easily seen to imply security of the {\sc and} computation.
Now, the secure computation of an {\sc and} gate---with statistical
security and using quantum communication---can be shown to require a
superpolynomial number of rounds if the adversary is
unbounded~\cite{NPS07}.  Therefore, the same must hold for a secure
password-based identification scheme.\footnote{In fact, we believe
  that the proof from~\cite{NPS07} can be extended to cover secure
  computation of equality of strings, which is equivalent to
  password-based identification. This would mean that we could prove
  the impossibility result directly, without the detour via a secure
  AND computation. 
}. In fact, in very recent work~\cite{BCS09prep}, using the definitions from~\cite{FS09},
it is  shown that the whole password of the honest player leaks to the dishonest player.

Another alternative approach is the classical bounded-storage model~\cite{Maurer90,CCM98,ADR02}. In contrast to our approach, only classical communication is used, and it is assumed that the attacker's {\em classical} memory is bounded. 
Unlike in the quantum case where we do not need to require the honest players to have any 
quantum memory, the classical bounded-storage model requires honest
parties to have a certain amount of memory which is related to the
allowed memory size of the adversary: if two legitimate users need $n$
bits of memory in an identification protocol meeting our security
criterion, then an adversary must be bounded in memory to $O(n^2)$
bits. The reason is that given a secure password-based identification
scheme, one can construct (in a black-box manner) a key-distribution
scheme that produces a one-bit key on which the adversary has an
(average) entropy of $\frac{1}{2}$.  On the other hand it is known
that in any key-distribution scheme which requires $n$ bits of memory
for legitimate players, an adversary with memory $\Omega(n^2)$ can obtain
the key except for an arbitrarily small amount of remaining
entropy~\cite{DM04}. It follows that password-based identification
schemes in the classical bounded-storage model can only be secure
against adversaries with memory at most $O(n^2)$.
This holds even for identification schemes with only passive 
security and without security against 
man-in-the-middle attacks. 
Roughly, the reduction works as follows. 
Alice and Bob agree on a public set of two keys $\{w_0,w_1\}$. Alice picks $a\in_R\{0,1\}$,
Bob picks $b\in_R\{0,1\}$, and they run the identification scheme
with keys $w_a$ and $w_b$ respectively.  The outcome
of the identification is then made public from which Bob
determines $a$. We argue that if the identification fails, i.e.\ $a \neq b$, then $a$ is a secure bit.  Thus, on average, $a$ has entropy (close to) $\frac{1}{2}$ from an eavesdropper's point of view.
Consider $w' \not\in \set{w_0,w_1}$. 
By the security
property of the identification scheme, Alice  and thus also a passive eavesdropper
Eve cannot distinguish between Bob having used $w_b$ or $w'$. Similarly, 
we can  then switch Alice's key $w_a$ to $w_{1-a}$ and Bob's switched key $w'$
to $w_{1-b}$ without changing Eve's view. 
Thus, Eve cannot distinguish an execution with $a = 0$ from one with $a = 1$ if $a \neq b$.

This limitation of the classical bounded-storage model is in sharp contrast 
with what we achieve in this paper, the honest players need no quantum memory at all
while our identification scheme remains secure against adversaries with 
quantum memory linear in the total number of qubits sent. The same separation
between the two models was shown for OT and bit commitment~\cite{DFSS05,DFRSS07}.

Finally, if one settles for the bounded-quantum-storage model, then in principle one could take a generic construction for general two-party secure-function-evaluation (SFE) based on OT together with the OT scheme from~\cite{DFSS05,DFRSS07} in order to implement a SFE for string equality and thus password-based identification. However, this approach leads to a highly impractical solution, as the generic construction requires many executions of OT, whereas our solution is comparable with {\em one} execution of the OT scheme from~\cite{DFSS05,DFRSS07}. Furthermore, SFE does not automatically take care of a man-in-the-middle attack, thus additional work would need to be done using this approach. 

\myparagraph{Subsequent Work.}  The difficulty of storing quantum
information can also be modeled differently from assuming a bound on
the physical number of qubits an adversary can control. In the more
realistic noisy-quantum-storage model put forward in~\cite{WST08}, all
incoming qubits can be stored by an adversary but are subject to
storage noise. Assuming a simple storage strategy, one can show that
the protocols in the current paper remain secure~\cite{STW08arxiv},
whereas it is unknown if security still holds in case of more 
sophisticated storage strategies~\cite{KWW09arxiv}.

If the storage limitation on the adversary fails to hold, it is easy
to see that not only will our security proofs fail, but in fact the
protocol we propose can be broken quite efficiently. However, it was
recently shown, in~\cite{DFLSS09full}, how to add a ``preamble'' to
the protocol using a commitment scheme based on a computational
assumption. It is shown in~\cite{DFLSS09full} that to break the
resulting protocol, an adversary must have both large quantum memory
{\em and} large computing power.

%%%%%%%%%%%%%%%%%%%%%%%%%%%%%%%%%%%%%%%%%%%%%%%%%%%%%%%%%%%%%%%%%
\section{Preliminaries}
%%%%%%%%%%%%%%%%%%%%%%%%%%%%%%%%%%%%%%%%%%%%%%%%%%%%%%%%%%%%%%%%%
%%%%%%%%%%%%%%%%%%%%%%%%%%%%%%%%%%%%%%%%%%%%%%%%%%%%%%%%%%%%%%%%%
\subsection{Notation and Terminology}
%%%%%%%%%%%%%%%%%%%%%%%%%%%%%%%%%%%%%%%%%%%%%%%%%%%%%%%%%%%%%%%%%

\myparagraph{Quantum States. } We assume the reader's familiarity with
basic notation and concepts of quantum information
processing~\cite{NC00}.  In this paper, the computational or
$+\,$-basis is defined by the pair $\{ \ket{0}, \ket{1} \}$ (also
written as $\{ \ket{0}_+, \ket{1}_+ \}$). The pair $\{ \ket{0}_\times,
\ket{1}_\times \}$ denotes the diagonal or $\times$-basis, where
$\ket{0}_\times = (\ket{0} + \ket{1}) / \sqrt{2}$ and $\ket{1}_\times
= (\ket{0} - \ket{1}) / \sqrt{2}$. We write $\ket{x}_\theta =
\ket{x_1}_{\theta_1} \otimes \cdots \otimes \ket{x_n}_{\theta_n}$ for
the $n$-qubit state where string $x = (x_1,\ldots,x_n) \in \{0,1\}^n$
is encoded in bases $\theta = (\theta_1,\ldots,\theta_n) \in
\{+,\times\}^n$.

The behavior of a (mixed) quantum state in a register $\regE$ is fully
described by its density matrix~$\rho_\regE$.  In order to simplify
language, we tend to be a bit sloppy and use $\regE$ as well as
$\rho_\regE$ as ``naming'' for the quantum state.  We often consider
cases where a quantum state $\regE$ may depend on some classical
random variable $X$ (from a finite set $\mX$) in that the state is
described by the density matrix $\rho_\regE^x$ if and only if $X =
x$. For an observer who has only access to the state $\regE$ but not
to $X$, the behavior of the state is determined by the density matrix
$\rho_{\regE} \assign \sum_x P_X(x) \rho_\regE^x$, whereas the joint
state, consisting of the classical $X$ and the quantum state $\regE$,
is described by the density matrix $\rho_{X\regE} \assign \sum_x
P_X(x) \proj{x} \otimes \rho_\regE^x$, where we understand
$\set{\ket{x}}_{x \in \mX}$ to be the standard (orthonormal)
basis of $\C^{|\mX|}$.
More general, for any event ${\cal E}$ (defined by $P_{\ev|X}(x) = P[\ev|X\!=\!x]$ for all $x$), we write 
\begin{equation}\label{eq:states}
\rho_{X\regE|{\cal E}} \assign \sum_x P_{X|{\cal E}}(x) \proj{x} \otimes \rho_\regE^x
\quad\text{and}\quad
\rho_{\regE|{\cal E}} \assign \tr_X(\rho_{X\regE|{\cal E}}) = \sum_x P_{X|{\cal E}}(x)\rho_\regE^x \enspace .
\end{equation}
We also write $\rho_X \assign \sum_x P_X(x) \proj{x}$ for the quantum representation of the classical random variable $X$ (and similarly for $\rho_{X|{\cal E}}$). 
This notation extends naturally to quantum states that depend on several classical random variables, defining the density matrices $\rho_{XY\regE}$, $\rho_{XY\regE|\ev}$, $\rho_{Y\regE|X=x}$ etc.
We tend to slightly abuse notation and write $\rho^x_{Y\regE} = \rho_{Y\regE|X=x}$ and $\rho^x_{Y\regE|\ev} = \rho_{Y\regE|X=x,\ev}$, as well as $\rho^x_{\regE} = \tr_Y(\rho^x_{Y\regE})$ and $\rho^x_{\regE|\ev} = \tr_Y(\rho^x_{Y\regE|\ev})$.\footnote{The density matrix $\rho^x_{\regE|\ev}$ describes the quantum state $\regE$ in the case that the event $\ev$ occurs and $X$ takes on the value $x$. The corresponding holds for the other density matrices considered here. } 
Note that writing $\rho_{X\regE} = \tr_Y(\rho_{XY\regE})$ and $\rho_{\regE} = \tr_{X,Y}(\rho_{XY\regE})$ is consistent with the above notation. We also write $\rho_{X\regE|\ev} = \tr_Y(\rho_{XY\regE|\ev})$ and $\rho_{\regE|\ev} = \tr_{X,Y}(\rho_{XY\regE|\ev})$, where one has to be aware that in contrast to \eqref{eq:states}, here the state $\regE$ may depend on the event $\ev$ (namely via $Y$), so that, e.g., $\rho_{\regE|\ev} = \sum_x P_{X|{\cal E}}(x)\rho_{\regE|\ev}^x$. 
Given a quantum state $\regE$ that depends on a classical random variable $X$, by saying that there exists a random variable $Y$ such that $\rho_{XY\regE}$ satisfies some condition, we mean that $\rho_{X\regE}$ can be understood as $\rho_{X\regE} = \tr_Y(\rho_{XY\regE})$ for some $\rho_{XY\regE}$ (with classical~$Y$) and that $\rho_{XY\regE}$ satisfies the required condition.\footnote{This is similar to the case of distributions of classical random variables where given $X$ the existence of a certain $Y$ is understood that there exists a certain joint distribution $P_{XY}$ with $\sum_y P_{XY}(\cdot,y) = P_X$. }

$X$ is independent of $\regE$ (in that $\rho^x_\regE$ does not depend on $x$) if and only if $\rho_{X\regE} = \rho_X \otimes \rho_\regE$, which in particular implies
that no information on $X$ can be learned by observing only $\regE$. 
Similarly, $X$ is random and independent of $\regE$ if and only if $\rho_{X\regE} = \frac{1}{|\mX|}\I \otimes \rho_\regE$, where $\frac{1}{|\mX|}\I$ is the density matrix of the fully mixed state of suitable dimension. 
Finally, if two states like $\rho_{X\regE}$
and $\rho_X \otimes \rho_\regE$ are $\eps$-close in terms of their
trace distance $\delta(\rho,\sigma) = \frac{1}{2} \tr(|\rho-\sigma|)$, which we write as $\rho_{X\regE} \epsclose \rho_X \otimes \rho_\regE$,
then the real system $\rho_{X\regE}$ ``behaves'' as the ideal
system $\rho_X \otimes \rho_\regE$ except with probability~$\varepsilon$ in that for any evolution of the system no
observer can distinguish the real from the ideal one with advantage
greater than~$\eps$~\cite{RK05}. 
As $\varepsilon$ can be
interpreted as an error probability, we typically require
$\varepsilon$ to be {\em negligible} in a security parameter $n$, denoted as $\eps = \negl(n)$.
A security parameter is a natural number $n$ given as input to all players in our protocols, and
a probability is said to be negligible in $n$ if for any polynomial $p$, it is smaller than $1/p(n)$ for all sufficiently large $n$.

\myparagraph{Conditional Independence. }
We also need to express that a random variable $X$ is (close to) independent of a quantum state 
$\regE$ {\em when given a random variable $Y$}. This means that when given $Y$, the state $\regE$ gives no (or little) additional information on $X$. Formally, this is expressed by requiring that $\rho_{X Y \regE}$ equals (or is close to) $\rho_{X\leftrightarrow Y \leftrightarrow \regE}$, which is defined as\footnote{The notation is inspired by the classical setting where the corresponding independence of $X$ and $Z$ given $Y$ can be expressed by saying that 
$X \leftrightarrow Y \leftrightarrow Z$ forms a Markov chain. }
$$
\rho_{X\leftrightarrow Y \leftrightarrow \regE} :=  \sum_{x,y}P_{X Y}(x,y)\proj{x} \otimes \proj{y} \otimes \rho_{\regE}^y \, .
$$ 
In other words, $\rho_{X Y \regE} = \rho_{X\leftrightarrow Y \leftrightarrow \regE}$ precisely if $\rho_\regE^{x,y} = \rho_\regE^{y}$ for all $x$ and $y$. 
To further illustrate its meaning, notice that if the $Y$-register is measured and value $y$ is obtained, then the
state $\rho_{X\leftrightarrow Y \leftrightarrow \regE}$ collapses to 
$ (\sum_{x}P_{X|Y}(x|y)\proj{x} )\otimes \rho_{\regE}^y$, so that indeed no further
information on $x$ can be obtained from the $\regE$-register. This notation naturally extends to $\rho_{X\leftrightarrow Y \leftrightarrow \regE|{\cal E}}$ simply by considering $\rho_{X Y \regE|{\cal E}}$ instead of $\rho_{X Y \regE}$. Explicitly, 
 $\rho_{X\leftrightarrow Y \leftrightarrow \regE|{\cal E}} = \sum_{x,y}P_{X Y|\ev}(x,y)\proj{x} \otimes \proj{y} \otimes \rho_{\regE|\ev}^{y}$.

 The notion of conditional independence has been introduced
 in~\cite{DFSS07} (a classical version was independently proposed
 in~\cite{CW08}) and used as a convenient tool in subsequent
 papers~\cite{FS09,BCS09prep}.  In this paper we will use the
 following property of conditional independence whose proof is given
 in Appendix~\ref{app:decomp}.
\begin{lemma}\label{lemma:decomp}
For any event $\ev$, the density matrix $\rho_{X\leftrightarrow Y \leftrightarrow \regE}$ can be decomposed into 
$$
\rho_{X\leftrightarrow Y \leftrightarrow \regE} = P[\ev]^2 \cdot \rho_{X\leftrightarrow Y \leftrightarrow \regE|{\cal E}} + (1-P[\ev]^2) \cdot \tau
$$ 
for some density matrix $\tau$. Furthermore, if $\ev$ is independent of $X$ and $Y$, then
$$
\rho_{X\leftrightarrow Y \leftrightarrow \regE} = P[\ev] \cdot \rho_{X\leftrightarrow Y \leftrightarrow \regE|{\cal E}} + P[\bar{\ev}] \cdot \rho_{X\leftrightarrow Y \leftrightarrow \regE|\bar{\cal E}} \, .
$$ 
\end{lemma}

\myparagraph{(Conditional) Smooth Min-Entropy. } 
Different notions of conditional (smooth) min-entropy have been proposed in the literature; we briefly specify here the variant that is convenient for us. 
Let $X$ and $Y$ be random variables, over respective
finite alphabets $\mX$ and $\mY$, with joint distribution $P_{XY}$. The
{\em conditional min-entropy} of $X$ given $Y$ is defined as the negative
logarithm of the guessing probability of $X$ given~$Y$: $\Hmin(X|Y)
\assign - \log \bigl( \guess(X|Y) \bigr)$ where 
$$
\guess(X|Y) \assign \sum_y P_Y(y) \max_x
P_{X|Y}(x|y) = \sum_y \max_x
P_{XY}(x,y) 
$$ 
and $\log$ denotes the binary logarithm (here and throughout the paper). 
More generally, we define $\Hmin(X\ev|Y)$ for any event $\ev$ as $\Hmin(X\ev|Y)
\assign - \log \big( \guess(X\ev|Y) \big)$ where%
\footnote{$\guess(X\ev|Y)$ can be understood as the optimal probability in guessing $X$ {\em and} have $\ev$ occur, when given~$Y$. }
$$
\guess(X\ev|Y) \assign \sum_y P_Y(y) \max_x P_{X\ev|Y}(x|y) = \sum_y \max_x P_{XY\ev}(x,y) \, .
$$ 
The \emph{conditional smooth min-entropy}
$\Hmin[\eps](X|Y)$ is then defined as
\[ \Hmin[\eps](X|Y) \assign \max_{\ev}
\Hmin(X \ev|Y) 
\]
where the max is over all events $\ev$ with $P[\ev] \geq 1 - \eps$.

Obviously, the unconditional versions of smooth and non-smooth min-entropy are obtained by using an ``empty'' $Y$; furthermore the above notions extend naturally to $\Hmin(X|Y,\ev)$ and $\Hmin[\eps](X|Y,\ev)$ for any event $\ev$ by considering the corresponding conditional joint distribution $P_{XY|\ev}$.

%%%%%%%%%%%%%%%%%%%%%%%%%%%%%%%%%%%%%%%%%%%%%%%%%%%%%%%%%%%%%%%%%
\subsection{Tools}
%%%%%%%%%%%%%%%%%%%%%%%%%%%%%%%%%%%%%%%%%%%%%%%%%%%%%%%%%%%%%%%%%

\myparagraph{Min-Entropy-Splitting . } A technical tool,
which will come in handy, is the following entropy-splitting lemma,
which may also be of independent interest.  Informally, it says that if for
a list of random variables, every pair has high (smooth) min-entropy,
then all of the random variables except one must have high (smooth)
min-entropy. The proof is given in Appendix~\ref{app:ES}. 

\begin{lemma}[Entropy-Splitting Lemma]\label{lemma:ES}
Let $\eps \geq 0$. Let $X_1,\ldots,X_m$ and $Z$ be random variables such that $\Hmin[\eps](X_i X_j|Z) \geq \alpha$ for all $i \neq j$. Then there exists a random variable $V$ over $\set{1,\ldots,m}$ such that for any {\em independent} random variable $W$ over $\set{1,\ldots,m}$ with $\Hmin(W) \geq 1$,
$$
\Hmin[2 m\eps](X_W|VWZ,V\!\neq\!W) \geq \alpha/2 - \log(m)
 - 1 \, .
$$
\end{lemma}

\myparagraph{Quantum Uncertainty Relation. } At the very core of our
security proofs lies (a special case of) the quantum uncertainty
relation from~\cite{DFRSS07}\footnote{In~\cite{DFRSS07}, a stricter
  notion of conditional smooth min-entropy was used, which in
  particular implies the bound as stated here. }, that lower bounds
the (smooth) min-entropy of the outcome when measuring an arbitrary
$n$-qubit state in a random basis $\theta \in \set{0,1}^n$.

\begin{theorem}[Uncertainty Relation~\cite{DFRSS07}]\label{thm:uncertainty}
Let $\regE$ be an arbitrary fixed $n$-qubit state. Let
$\Theta$ be uniformly distributed over $\set{+,\times}^n$ (independent
of $\regE$), and let $X\in\{0,1\}^n$ be the random variable for the outcome of
measuring $\regE$ in basis~$\Theta$.
Then, for any $\lambda > 0$, the conditional smooth min-entropy is lower bounded by
$$
\Hmin[\eps](X|\Theta) \geq \Big(\frac12 - 2\lambda\Big) n 
$$
with $\eps \leq 2^{-\sigma(\lambda)n}$ and $\sigma(\lambda) = \frac{\lambda^2 \log(e)}{32(2-\log(\lambda))^2}$. 
\end{theorem}
Thus, ignoring negligibly small ``error probabilities'' and linear fractions that can be chosen arbitrarily small, the outcome of measuring any $n$-qubit state in a random basis has $n/2$ bits of min-entropy, given the basis.

\myparagraph{Privacy Amplification. }
Finally, we recall the quantum-privacy-amplification theorem of Renner and K\"onig~\cite{RK05}. 
The version we use here follows immediately from~\cite[Corollary
5.6.1]{Renner05} by applying the chain rule for min- and
max-entropy~\cite[Lemma~3.2.9]{Renner05} and using the equivalence, as
shown in~\cite{KRS08arxiv}, of the quantum and the classical notion of
(smooth) conditional min-entropy.  Recall that a class $\cal F$ of
hash functions from $\cal X$ to $\cal Y$ is called (strongly)
universal-2 if for any $x \neq x' \in \cal X$, and for $F$ uniformly
distributed over $\cal F$, the collision probability $P[F(x) \!=\!
F(x')]$ is upper bounded by $1/|{\cal Y}|$, respectively, for the
strong notion, the random variables $F(x)$ and $F(x')$ are uniformly
and independently distributed over $\cal Y$.

\begin{theorem}\label{thm:pa}
Let $X$ and $Z$ be random variables distributed over $\cal X$ and $\cal Z$, respectively, and let $\regE$ be a $q$-qubit state that may depend on $X$ and $Z$. Let $F$ be the random and independent choice of a member of a universal-2 class of hash functions ${\cal F}$ from $\cal X$ into $\set{0,1}^{\ell}$. 
Then, for any $\varepsilon > 0$
$$
  \delta\big(\rho_{F(X) F Z \regE}, {\textstyle\frac{1}{2^{\ell}}\I } \otimes \rho_{F Z \regE}\big) \leq \frac{1}{2} \, 2^{-\frac{1}{2}\big(\Hmin[\eps](X|Z)-q-\ell\big)} + 2 \eps \, . \label{dbound}
$$
\end{theorem}

%%%%%%%%%%%%%%%%%%%%%%%%%%%%%%%%%%%%%%%%%%%%%%%%%%%%%%%%%%%%%%%%%
\section{The Identification Scheme}
%%%%%%%%%%%%%%%%%%%%%%%%%%%%%%%%%%%%%%%%%%%%%%%%%%%%%%%%%%%%%%%%%

%%%%%%%%%%%%%%%%%%%%%%%%%%%%%%%%%%%%%%%%%%%%%%%%%%%%%%%%%%%%%%%%%
\subsection{The Setting}
%%%%%%%%%%%%%%%%%%%%%%%%%%%%%%%%%%%%%%%%%%%%%%%%%%%%%%%%%%%%%%%%%

We assume that the honest user $\U$ and the honest server $\S$ share
some key $w \in \W$  (which we think of as a password), where the
choice of $w$ is described by the random variable $W$.
An identification protocol is now simply any protocol for $\U$ and $\S$ 
using classical and/or quantum communication where the parties are both given
as input a security parameter $n$ and (in the honest case) the password $w$, and 
where $\S$ outputs accept or reject in the end. 

We do not require $\W$ to be very large (i.e.  $|\W|$ does not have to
be lower bounded by the security parameter in any way), and $w $ does
not necessarily have to be uniformly distributed in $\W$. So, we may
think of $w $ as a human-memorizable password or PIN code. The goal of
this section is to construct an identification scheme that allows $\U$
to ``prove'' to $\S$ that he knows $w$. The scheme should have the
following security properties: a dishonest server $\dS$ learns
essentially no information on $w$ beyond that he can come up with a
guess $w'$ for $w$ and learns whether $w' = w$ or not, and similarly a
dishonest user succeeds in convincing the verifier essentially only if
he guesses $w$ correctly, and if his guess is incorrect then the only
thing he learns is that his guess is incorrect. This in particular
implies that as long as the entropy of $W$ is large enough, the
identification scheme may be safely repeated.  Finally, it must of
course be the case that $\S$ accepts the legitimate user who has the
correct password.  More formally, we require the following:

\begin{definition} \label{def:correctness}
An execution by honest
$\U,\S$ on input $w$ for both parties results in $\S$ accepting, except with negligible probability
(as a function of $n$).
\end{definition}

\begin{definition} \label{def:usersecurity}
We say that an identification protocol for two parties $\U, \S$ is secure for the user with
error $\eps$ against (dishonest) server $\dS$ if the following is
satisfied: whenever the initial state of $\dS$ is independent of $W$,
the joint state $\rho_{W \regS}$ after the execution of the protocol is such that there exists a random variable $W'$ that is independent of $W$ and such that 
$$
\rho_{W W' \regS|W'\neq W} \close{\eps} \rho_{W\leftrightarrow W'\leftrightarrow  \regS|W'\neq W}.
$$
\end{definition}

\begin{definition} \label{def:serversecurity} We say that an
  identification protocol for two parties $\U, \S$ is secure for the
  server with error $\eps$ against (dishonest) user $\dU$ if the
  following is satisfied: whenever the initial state of a dishonest
  user $\dU$ is independent of $W$, there exists $W'$ (possibly
  $\perp$), independent of $W$, such that if $W \neq W'$ then $\S$
  accepts with probability at most $\eps$, and if $W = W'$ then $\S$ accepts
  with certainty. Furthermore, the common state $\rho_{W \regU}$ after
  the execution of the protocol (including $\S$'s announcement to
  accept or reject) satisfies
$$
\rho_{W W' \regU|W' \neq W} \close{\eps} \rho_{W\leftrightarrow W'\leftrightarrow \regU| W'\neq W} \enspace .
$$ 

\end{definition}

If these definitions are satisfied for a small $\eps$, we  are guaranteed 
that whatever a dishonest party does is
essentially as good as trying to guess $W$ by some arbitrary (but
independent) $W'$ and learning whether the guess was correct or not,
but nothing beyond that. Such a property is obviously the best one can hope for, since an attacker may always honestly execute the protocol with a guess for $W$ and observe whether the protocol was successful. 

We would like to point out that the above security definitions, and in
fact any security claim in this paper, guarantees {\em sequential
  self-composability}, as the output state is guaranteed to have the
same independency property (for any fixed choice of $W'$) as is
required from the input state (except if the attacker guesses
$W$). Moreover, it is shown in~\cite{FS08arxiv,FS09} that our
definitions imply a ``real/ideal'' world definition given
in~\cite{FS09}. More specifically, it is shown that a protocol
satisfying our information theoretic conditions implements a natural
ideal identification functionality, and by the composition theorem from~\cite{FS09},
this means that the protocol composes sequentially in a classical
environment, i.e.~the quantum protocol can be treated as the ideal
functionality when analyzing a more complicated classical
outer protocol.

It should be noted that security for user and server is usually not sufficient for application in practice of an identification protocol. A problem occurs if the honest user and server are interacting and an attacker can manipulate the communication, i.e., do a ``man-in-the-middle'' attack, and observe the reaction of the honest parties. This scenario is not covered by the above definitions, and indeed it turns out that the simplest version of our protocol is not secure against such an attack. Nevertheless, the problem can be solved and we address it  in Section \ref{mim}.

%%%%%%%%%%%%%%%%%%%%%%%%%%%%%%%%%%%%%%%%%%%%%%%%%%%%%%%%%%%%%%%%%
\subsection{The Intuition}
%%%%%%%%%%%%%%%%%%%%%%%%%%%%%%%%%%%%%%%%%%%%%%%%%%%%%%%%%%%%%%%%%

The scheme we propose is related to the (randomized) 1-2 OT scheme
of~\cite{DFRSS07}. In that scheme, Alice sends $\ket{x}_\theta$ to
Bob, for random $x \in \set{0,1}^n$ and $\theta \in \set{+,\times}^n$.
Bob then measures everything in basis $+$ or $\times$, depending on
his choice bit $c$, so that he essentially knows half of $x$ (where
Alice used the same basis as Bob) and has no information on the other
half (where Alice used the other basis), though, at this point, 
he does not know yet which bits he knows and which ones he does not. 
Then, Alice sends $\theta$
and two hash functions to Bob, and outputs the hash values $s_0$ and
$s_1$ of the two parts of $x$, whereas Bob outputs the hash value
$s_c$ that he is able to compute from the part of $x$ he knows. It is
proven in~\cite{DFRSS07} that no dishonest Alice can learn $c$, and for
any quantum-memory-bounded dishonest Bob, at least one of the two
strings $s_0$ and $s_1$ is random for Bob.

This scheme can be extended by giving Bob more options for measuring
the quantum state. Instead of measuring all qubits in the $+$ or the
$\times$ basis, he may measure using $m$ different strings of bases, where any
two possible basis-strings have large Hamming distance. Then Alice
computes and outputs $m$ hash values, one for each possible basis-string that
Bob might have used. She reveals $\theta$ and the hash functions
to Bob, so he can compute the hash value
corresponding to the basis that he has used, and no other hash
value. Intuitively, such an extended scheme leads to a randomized 1-$m$ OT. 

The scheme can now be transformed into a secure identification scheme as
follows, where we assume (wlog) that $\W = \set{1,\ldots,m}$. 
The user $\U$, acting as Alice, and the server $\S$, acting as
Bob, execute the randomized 1-$m$ OT scheme where $\S$
``asks'' for the string indexed by his key $w$, such that $\U$ obtains 
random strings $s_1,\ldots,s_m$ and $\S$ obtains~$s_w$. Then, to
do the actual identification, $\U$ sends
$s_w$ to $\S$, who accepts if and only if it coincides with his
string $s_w$. Intuitively, such a construction is secure against a dishonest server
since unless he asks for the right string (by guessing $w $ correctly)
the string $\U$ sends him is random and thus gives no information on
$w $. On the other hand, a dishonest user does not know which of the
$m$ strings $\S$ asked for and wants to see from him.
We realize this intuitive idea in the next section. In the actual protocol,
$\U$ does not have to explicitly compute all the $s_i$'s,
and also we only need a single hash function (to compute $s_w$).
We also take care of some subtleties, for instance that the $s_i$ are not necessarily
random if Alice (i.e. the user) is dishonest.

%%%%%%%%%%%%%%%%%%%%%%%%%%%%%%%%%%%%%%%%%%%%%%%%%%%%%%%%%%%%%%%%%
\subsection{The Basic Scheme}
%%%%%%%%%%%%%%%%%%%%%%%%%%%%%%%%%%%%%%%%%%%%%%%%%%%%%%%%%%%%%%%%%

Let $\code:\W\rightarrow\set{+,\times}^n$ be the encoding function of a binary code of length $n$ 
with $m = |\W|$ codewords and minimal distance $d$. $\code$ can be chosen such that $n$ is linear in $\log(m)$ or larger, and $d$ is linear in $n$.
Furthermore, let ${\cal F}$ and ${\cal G}$ be strongly universal-2 classes of hash functions\footnote{Actually, we only need $\cal G$ to be {\em strongly} universal-2. }
from $\set{0,1}^n$ to $\set{0,1}^{\ell}$ and from $\W$ to $\set{0,1}^{\ell}$, respectively, for some parameter~$\ell$. For $x \in \set{0,1}^n$ and $I \subseteq \set{1,\ldots,n}$, we define $x|_I \in \set{0,1}^n$ to be the restriction of $x$ to the coordinates $x_i$ with $i \in I$. If $|I| < n$ then applying $f \in {\cal F}$ to $x|_I$ is to be understood as applying $f$ to $x|_I$ padded with sufficiently many $0$'s.

\begin{myfigure}{H}
\begin{myprotocol}{\QID}
\item\label{it:prepare} $\U$ picks $x \in_R \set{0,1}^n$ and $\theta \in_R \{+,\times \}^n$, and sends state $\ket{x}_\theta$ to $\S$. 
\item\label{it:measure} $\S$ measures $\ket{x}_\theta$ in basis $c = \code(w)$. Let $x'$ be the outcome. 
 \item\label{it:bound_applies} 
$\U$ picks $f \in_R {\cal F}$ and sends $\theta$ and $f$ to $\S$. Both compute $I_w \assign  \Set{i}{\theta_i \!=\! \code(w)_i}$. 
 \item $\S$ picks $g \in_R {\cal G}$ and sends $g$ to $\U$. 
 \item $\U$ computes and sends $z \assign f(x|_{I_w}) \oplus g(w)$ to $\S$. 
 \item $\S$ accepts if and only if $z = z'$ where $z' \assign f(x'|_{I_w}) \oplus g(w)$. 
\end{myprotocol}
\end{myfigure}

It is trivial that the protocol satisfies  Definition~\ref{def:correctness}. In addition, we have:

\begin{proposition}[User security]\label{prop:Usec} 
  Assume that the size of the quantum memory of dishonest server $\dS$
  is at most $q$ at step~\ref{it:bound_applies} of \QID, and that
  $\Hmin(W) \geq 1$. Then \QID is secure for the user with error
  $\eps$ against $\dS$ according to Definition~\ref{def:usersecurity},
  where
$$\eps = 2^{-\frac12((\frac14 - \lambda) d - \log(m) - q - \ell - 1)} + 2^{-(\sigma(\lambda) d - \log(m) - 3)}$$ 
for an arbitrary $0 < \lambda < \frac14$. 
\end{proposition}
Note that $\sigma(\lambda)$ was defined earlier in the claim of the uncertainty relation. To understand what the result on $\eps$ means, note that using a family of asymptotically good codes, we can assume that $d$ grows linearly with the main security parameter $n$, while still allowing $m$ (the number of passwords) to be exponential in $n$.  So we may choose the parameters such that 
$\frac{d}{n}, \frac{\log(m)}{n}, \frac{q}{n}$ and  $\frac{\ell}{n}$ are all constants. The result above now says that $\eps$ is exponentially small as a function of $n$ if these constants are chosen in such a way that for some $0< \lambda < \frac{1}{4}$, it holds that
$(\frac14-\lambda)\frac{d}{n} - \frac{ \log(m)}{n} - \frac{q}{n} - \frac{\ell}{n} > 0$ and
 $\sigma(\lambda)\frac{d}{n} -  \frac{\log(m)}{n}> 0$.
See Theorem \ref{thm:impersonation} for a choice of parameters that also take server security into account. 
If we are willing to assume that $\log(m)$ is sublinear in $n$, which may be quite reasonable is case we use short passwords that humans can remember, the condition further simplifies to
$\frac{d}{4n} - \frac{q}{n} - \frac{\ell}{n} > 0$.

\begin{proof}

We consider and analyze a {\em purified} version of \QID where in
step~\ref{it:prepare}, instead of sending $\ket{x}_\theta$ to $\dS$ for a
random $x$, $\U$ prepares a fully entangled state $2^{-n/2}\sum_x
\ket{x}\ket{x}$ and sends the second register to $\dS$ while
keeping the first. Then, in step~\ref{it:bound_applies} when the memory
bound has applied, he measures his register in the random basis
$\theta \in_R \set{+,\times}^n$ in order to
obtain $x$. Standard arguments imply that this purified version
produces exactly the same common state, consisting of the classical
information on $\U$'s side and $\dS$'s quantum state.

Recall that before step~\ref{it:bound_applies} is executed, the memory
bound applies to $\dS$, meaning that $\dS$ has to measure all but
$q$ of the qubits he holds, which consists of his initial state and
his part of the EPR pairs. Before doing the measurement, he may append
an ancilla register and apply an arbitrary unitary transform. As a
result of $\dS$'s measurement, $\dS$ gets some outcome $y$, and the
common state collapses to a $(n+q)$-qubit state (which depends on
$y$), where the first $n$ qubits are with $\U$ and the remaining $q$
with $\dS$. The following analysis is for a fixed $y$, and works no
matter what $y$ is.

We use upper case letters $W$, $X$, $\Theta$, $F$, $G$ and $Z$ for the
random variables that describe the respective values $w$, $x$,
$\theta$ etc. in an execution of the purified version of \QID. We
write $X_j = X|_{I_j}$ for any $j$, and we let $\regS'$ be $\dS$'s
$q$-qubit state at step~\ref{it:bound_applies}, after the
memory bound has applied. Note that $W$ is independent of $X$,
$\Theta$, $F$, $G$ and $\regS'$.

For $1 \leq i \neq j \leq m$, 
fix the value of $X$, and correspondingly of $X_i$ and $X_j$, at the positions where $\code(i)$ and 
$\code(j)$ coincide, and focus
on the remaining (at least) $d$ positions. The uncertainty
relation (Theorem~\ref{thm:uncertainty}) implies that the restriction of $X$ to
these positions has $(\frac12 - 2\lambda)d$ bits of $\eps'$-smooth min-entropy
given~$\Theta$, where $\eps' \leq 2^{-\sigma(\lambda)d}$ and $0< \lambda < \frac12$ arbitrary. Since every bit in the restricted $X$ appears in one of $X_i$ and $X_j$, the pair $X_i,X_j$ also has $(\frac12 - 2\lambda)d$ bits of $\eps'$-smooth min-entropy
given~$\Theta$. 
The Entropy Splitting Lemma~\ref{lemma:ES} implies that there exists $W'$ (called $V$ in Lemma~\ref{lemma:ES}) such that if $W \neq W'$ then $X_W$ has $(\frac14 - \lambda)d - \log(m) - 1$ bits of $2m\eps'$-smooth min-entropy given $W$ and $W'$ (and~$\Theta$). Privacy amplification then guarantees that $F(X_W)$ is $\eps''$-close to random and independent of $F, W, W',\Theta$ and $\regS'$, conditioned on $W \neq W'$, where \smash{$\eps'' = \frac12\cdot 2^{-\frac12(d/4 - \lambda d - \log(m) - 1 - q - \ell)} + 4m\eps'$}. 
It follows that $Z = F(X_W) \oplus G(W)$ is $\eps''$-close to random and independent of $F, G, W, W',\Theta$ and $\regS'$, conditioned on $W \neq W'$. 

Formally, we want to upper bound
$\delta(\rho_{W W' \regS| W'\neq W},\  \rho_{W\leftrightarrow W'\leftrightarrow  \regS| W'\neq W})$. 
Since the output state $\regS$ is, without loss of generality, 
obtained by applying some unitary transform to
the set of registers
$(Z,F,G,W',\Theta, \regS')$, the distance above is equal to 
the distance between 
$\rho_{W W' (Z,F,G,\Theta, \regS')| W'\neq W}$ and $\rho_{W\leftrightarrow W'\leftrightarrow  (Z,F,G,\Theta, \regS')| W'\neq W}$. 
We then get:
\begin{align*}
\rho&_{W W' (Z, F, G,  \Theta, \regS')|W'\neq W} \close{\eps''} 
{\textstyle\frac{1}{2^\ell}}\I \otimes \rho_{W W' (F, G,  \Theta, \regS')|W'\neq W} \\
&=\; {\textstyle\frac{1}{2^\ell}}\I \otimes \rho_{W\leftrightarrow W'\leftrightarrow (F,G,\Theta, \regS')| W'\neq W} 
\close{\eps''}  \rho_{W\leftrightarrow W'\leftrightarrow  (Z,F,G,\Theta, \regS')| W'\neq W} \enspace ,
\end{align*}
where approximations follow from privacy amplification and the exact equality
comes from the independency of $W$, which, when conditioned 
on  $W' \neq W$, translates to independency given $W'$. 
The claim follows with $\eps = 2\eps''$. 
\qed
\end{proof}

\begin{proposition}[Server security]\label{prop:Ssec}
  If $\Hmin(W) \geq 1$, then \QID is secure for the server with error
  $\eps$ against any $\dU$ according to
  Definition~\ref{def:serversecurity}, where $\eps= m^2/2^\ell$.
\end{proposition}
The formal proof is given
below. 
The idea is the following. We let $\dU$ execute \QID with a server that is {\em unbounded} in quantum memory. Such a server can obviously obtain $x$ and thus compute $s_j = f(x|_{I_j}) \oplus g(j)$ for all $j$. Note that $s_w$ is the message $z$ that $\dU$ is required to send in the last step. Now, if the $s_j$'s are all distinct, then $z$ uniquely defines $w'$ such that $z = s_{w'}$, and thus $S$ accepts if and only if $w' = w$, and $\dU$ does not learn anything beyond. The strong universal-2 property of $g$ guarantees that the $s_j$'s are all distinct except with probability $m^2/2^{\ell}$. 

\begin{proof}
Again, we consider a slightly modified version. We let $\dU$ interact
with a server that has {\em unbounded} quantum memory and does the
following. Instead of measuring $\ket{x}_\theta$ in step~\ref{it:measure} in basis $c$, it stores the state and measures it after step~\ref{it:bound_applies} in basis $\theta$ (and obtains $x$). 
This modified version produces the same common state $\rho_{W \regU}$ as the original scheme, since the only difference between the two is when and in what basis the qubits at positions $i \not\in I_w$ are measured, which does not effect the execution in any way. 

We use the upper case letters $W$, $X$, $\Theta$, $F$, $G$ and $Z$ for the random variables that describe the respective values $w$, $x$, $\theta$ etc. in an execution of the modified version of \QID. 
Furthermore, we define $S_j \assign F(X|_{I_j}) \oplus G(j)$ for $j=1,\ldots,m$. Note that $Z' = S_W$ represents the value $z'$ used by $\S$ in the last step. 
Let $\ev$ be the event that all $S_j$'s are distinct. 
By the strong universal-2 property, and since $G$ is independent of $X$ and $F$, the $S_j$'s are pairwise independent and thus it follows from the union bound that $\ev$ occurs except with probability at most $m(m-1)/2 \cdot 1/2^{\ell} \leq m^2/2^{\ell+1}$. 

Let $\regU'$ be $\dU$'s quantum state after the execution of \QID but
{\em before} he learns $\S$'s decision to accept or reject. We may
assume that the values of all random variables $X$, $\Theta$, $F$,
$G$, $Z$ and the $S_j$'s are known/given to $\dU$, i.e., we consider
them as part of $\regU'$.  Furthermore, we may assume that $Z$ is one
of the $S_j$'s, i.e. that $Z = S_{W'}$ for a random variable
$W'$. Indeed, if $Z \neq S_j$ for all $j$ then we set $W' \assign
\perp$ and $\S$'s decision is ``reject'', no matter what $W$ is, and
$\dU$ obviously learns no information on $W$ at all. By the way we
have defined $W'$, is clear that $\S$ accepts if $W=W'$.

Note that $\regU'$ is independent of $W$ by assumption on $\dU$'s
initial state (in Definition~\ref{def:serversecurity}) and by
definition of the random variables $X$, $\Theta$ etc. Since $\ev$ is
determined by
the $S_j$'s (which are part of $\regU'$), this holds also when
conditioning on $\ev$. This then translates to the independence of
$\regU'$
from $W$ when given $W'$, conditioned on $W' \neq W$ and~$\ev$.

We now consider $\dU$'s state $\regU$ {\em after} he has learned
$\S$'s decision. If $W' \neq W$ and all $S_j$'s are distinct then $\S$
rejects with probability $1$. Hence, conditioned on the events $W'
\neq W$ and $\ev$, $\dU$'s state $\regU$
remains independent of $W$ given $W'$.
Define $p:= P[\ev| W'\!\neq\! W]$ and $\bar{p}:= P[\bar{\ev}| W'\!\neq\! W]$ $= 1-p$, where $\bar{\ev}$ is the complementary event to $\ev$.
Recall that $P[\bar{\ev}] \leq m^2/2^{\ell+1}$, and therefore  $\bar{p} \leq P[\bar{\ev}]/(1-P[W'\!=\!W]) \leq 2 P[\bar{\ev}] \leq m^2/2^{\ell}$, where the second-last inequality follows from the independence of $W$ and $W'$, and from the condition on $\Hmin(W)$. 
Note that $\bar{p}$ upper bounds the probability that $\S$ accepts in case $W' \neq W$, proving the first claim. 
From the above
it follows that
\begin{align*}
\rho_{W W' \regU|W' \neq W} 
&= p\cdot\rho_{W W' \regU|\ev,W' \neq W} + \bar{p}\cdot\rho_{W W' \regU|\bar{\ev},W' \neq W} \\
&= p\cdot\rho_{W\leftrightarrow W' \leftrightarrow \regU|\ev,W' \neq W} + \bar{p}\cdot\rho_{W W' \regU|\bar{\ev},W' \neq W} \, .
\end{align*}
Furthermore, it is not too hard to see that $\ev$ is independent of $W$ and $W'$, and thus also when conditioned on $W' \neq W$. Lemma~\ref{lemma:decomp} hence implies that 
\begin{align*}
\rho_{W\leftrightarrow W' \leftrightarrow \regU|W' \neq W} 
&= p\cdot\rho_{W\leftrightarrow W' \leftrightarrow \regU|\ev,W' \neq W} + \bar{p}\cdot\rho_{W\leftrightarrow W' \leftrightarrow \regU|\bar{\ev},W' \neq W} \,. 
\end{align*}
By definition of the metric $\delta(\cdot,\cdot)$, and because it cannot be bigger than 1, the distance between the two states is at most $\bar{p} \leq m^2/2^{\ell}$. 
\qed
\end{proof}

We call an identification scheme {\em $\eps$-secure against impersonation attacks} if the protocol is secure for the user and secure for the sender with error at most $\eps$ in both cases. The following holds:

\begin{theorem}
\label{thm:impersonation}
If $\Hmin(W) \geq 1$, then the identification scheme \QID (with suitable choice of parameters) is $\eps$-secure against impersonation attacks for any unbounded user and for any server with quantum memory bound $q$, where 
$$
\eps = 2^{-\frac13((\frac14 - \lambda) n\mu - 3\log(m) - q - 2)} + 
2^{-(\sigma(\lambda) n\mu - \log(m) - 4)}
$$ 
for an arbitrary $0 < \lambda < \frac14$, and where $\mu =
h^{-1}(1-\log(m)/n)$, and $h^{-1}$ is the inverse function of the
binary entropy function: $h(p) \assign
-p\cdot\log(p)-(1-p)\cdot\log(1-p)$ restricted to $0 < p \leq
\frac12$.  In particular, if $\log(m)$ is sublinear in $n$, then
$\eps$ is negligible in $n - 8q$.
\end{theorem}

\begin{proof}  %
  We choose $\ell = \frac13\big((\frac14-\lambda)d +3 \log(m) - q - 1
  \big)$. Then user security holds except with an error $\eps =
  2^{-\frac13((\frac14 - \lambda) d - 3\log(m) - q - 1)} +
  2^{-(\sigma(\lambda) d - 2\ln(m) - 3)}$, and server security holds
  except with an error $m^2/2^\ell = 2^{-\frac13((\frac14 - \lambda) d
    - 3\log(m) - q - 1)}$. Using a code $\code$, which asymptotically
  meets the Gilbert-Varshamov bound~\cite{Tho83}, $d$ may be chosen
  arbitrarily close to $n \cdot h^{-1}\bigl(1-\log(m)/n\bigr)$.  In
  particular, we can ensure that $d$ differs from this value by at
  most 1. Inserting $d= n \cdot h^{-1}\bigl(1-\log(m)/n\bigr) -1$ in
  the expression for user security yields the theorem.  \qed
\end{proof}

%%%%%%%%%%%%%%%%%%%%%%%%%%%%%%%%%%%%%%%%%%%%%%%%%%%%%%%%%%%%%%%%%
\subsection{Mutual Identification}
%%%%%%%%%%%%%%%%%%%%%%%%%%%%%%%%%%%%%%%%%%%%%%%%%%%%%%%%%%%%%%%%%
 
In order to obtain {\em mutual} identification, where also the server
identifies himself towards the user, one could of course simply run
\QID in both directions: say, first $\U$ identifies himself to $\S$,
and then $\S$ identifies himself to $\U$ (by exchanging their roles in
\QID).  However, this scheme allows the dishonest server to exclude
{\em two} possible keys $w \in {\cal W}$ per invocation, and it
requires to also assume the {\em user's} quantum memory to be bounded,
and has doubled complexity.

We briefly sketch an approach that circumvents these drawbacks of the
trivial solution: In the original \QID scheme, instead of announcing
$z = f(x|_{I_w}) \oplus g(w)$, $\U$ announces a {\em noisy version}
$\tilde{z}$, obtained from $z$ by flipping each bit of $z$
independently with some small probability; this still allows $\S$ to
verify if $\U$ knows $w$ by testing if $\tilde{z}$ is ``close" to
$z'$, and $\S$ has then to prove knowledge of $w$ by announcing to
$\U$ the positions where $\U$ flipped the bits.

Security against a dishonest user still holds (with a slightly larger error probability) since the uniformity of the $S_j$'s, as defined in the proof, also guarantees that the $S_j$'s are pair-wise ``far apart" so that $W'$ is still uniquely determined by $\tilde{Z}$. And security against a dishonest server follows from the fact that if $W' \neq W$ then $Z$ is (essentially) uniformly distributed and thus given its noisy version 
$\tilde{Z}$ the server can at best guess the positions of the bit-flips, which are independent of $W$.

%%%%%%%%%%%%%%%%%%%%%%%%%%%%%%%%%%%%%%%%%%%%%%%%%%%%%%%%%%%%%%%%%
\subsection{An Error-tolerant Scheme}\label{sec:noise}
%%%%%%%%%%%%%%%%%%%%%%%%%%%%%%%%%%%%%%%%%%%%%%%%%%%%%%%%%%%%%%%%%

We now consider an imperfect quantum channel with ``error rate''
$\phi$. The scheme \QID is sensitive to such errors in that they cause
$x|_{I_w}$ and $x'|_{I_w}$ to be different and thus an honest server
$\S$ is likely to reject an honest user $\U$. This problem can be
overcome by means of error-correcting techniques: $\U$ chooses a
linear error-correcting code that allows to correct a $\phi$-fraction
of errors, and then in step~\ref{it:measure}, in addition to $\theta$
and $f$, $\U$ sends a description of the code and the syndrome $s$ of
$x|_{I_w}$ to $\S$; this additional information allows $\S$ to recover
$x|_{I_w}$ from its noisy version $x'|_{I_w}$ by standard techniques.
However, this technique introduces a new problem: the syndrome $s$ of $x|_{I_w}$
may give information on $w $ to a dishonest server. Hence, to
circumvent this problem, the code chosen by $\U$ must have the
additional property that for a dishonest user, who has high
min-entropy on $x|_{I_w}$, the syndrome $s$ is (close to) independent
of $w$.

This problem has been addressed and solved in the classical setting by
Dodis and Smith~\cite{DS05}, and subsequently in the quantum
setting in~\cite{FS08}. Dodis and Smith present a family of efficiently
decodable linear codes allowing to correct a constant fraction of
errors, and where the syndrome of a string is close to uniform if the
string has enough min-entropy and the code is chosen at random from
the family.
Specifically, Lemma~5 of~\cite{DS05} guarantees that for every $0 < \lambda < 1$ and for an infinite number of $n'$'s there exists a {\em $\delta$-biased} (as defined in~\cite{DS05}) family ${\cal C} = \set{C_j}_{j \in \cal J}$ of $[n',k',d']_2$-codes with $\delta < 2^{-\lambda n'/2}$, and which allows to efficiently correct a constant fraction of errors. 
Furthermore, Theorem~3.2 of \cite{FS08} (which generalizes Lemma~4 in~\cite{DS05} to the quantum setting) guarantees that if a string $Y$ has $t$ bits of min-entropy\footnote{\cite{FS08} does not consider {\em smooth} min-entropy, but it is not too hard to see that their results also hold for the smooth version. } 
then for a randomly chosen code $C_j \in \cal C$, the syndrome of $Y$ is close to random and independent of $j$ and any $q$-qubit state that may depend on $Y$, where the closeness is given by 
$\delta \cdot 2^{(n' + q - t)/2}$. 
In our application, $Y = X_W$, $n' \approx n/2$ and $t \approx d/4 - \log(m) - \ell$, where the additional loss of $\ell$ bits of entropy comes from learning the $\ell$-bit string $z$.
Choosing $\lambda = 1-\frac{t}{2 n'}$ gives an ensemble of code families that allow to correct a linear number of errors and the syndrome is $\eps$-close to uniform given the quantum state, where $\eps \leq 2^{-n'/2+t/4} \cdot 2^{(n' + q - t)/2} = 2^{-(t-2q)/4}$, which is exponentially small provided that there is a linear gap between $t$ and $2q$. 
Thus, the syndrome gives essentially no additional information. 
The error rate $\phi$ that can be tolerated this way depends in a rather complicated way on $\lambda$, but choosing $\lambda$ larger, for instance $\lambda = 1-\frac{t + \nu q}{2 n'}$ for a constant $\nu > 0$, allows to tolerate a higher error rate but requires $q$ to be a smaller (but still constant) fraction of $t$.

Another imperfection has to be taken into account in current 
implementations of the quantum channel: imperfect sources. An imperfect
source transmits more than one qubit in the same state with probability $\eta$
independently each time a new transmission takes place. 
To deal with imperfect sources, we freely give away $(x_i,\theta_i)$ to the adversary when
a multi-qubit transmission occurs in position $i$. It is not difficult to see that
parameter $\varepsilon$ in Proposition~\ref{prop:Usec} then changes in that $d$ is replaced by $(1-\eta)d$.  

It follows that
a quantum channel with error-rate $\phi$ and multi-pulse rate $\eta$, called
the $(\phi,\eta)$-weak quantum model in~\cite{DFSS05}, can be tolerated
for some small enough (but constant) $\phi$ and~$\eta$.

%%%%%%%%%%%%%%%%%%%%%%%%%%%%%%%%%%%%%%%%%%%%%%%%%%%%%%%%%%%%%%%%%
\section{Defeating Man-in-the-Middle Attacks}
\label{mim}
%%%%%%%%%%%%%%%%%%%%%%%%%%%%%%%%%%%%%%%%%%%%%%%%%%%%%%%%%%%%%%%%%

%%%%%%%%%%%%%%%%%%%%%%%%%%%%%%%%%%%%%%%%%%%%%%%%%%%%%%%%%%%%%%%%%
\subsection{The Approach}\label{sec:xapproach}
%%%%%%%%%%%%%%%%%%%%%%%%%%%%%%%%%%%%%%%%%%%%%%%%%%%%%%%%%%%%%%%%%

In the previous section, we ``only'' proved security against impersonation attacks, but we did not consider a man-in-the-middle attack, where the attacker sits between an honest user and an honest server and controls their (quantum and classical) communication. And indeed, \QID is highly insecure against such an attack: the attacker may measure the first qubit in, say, basis $+$, and then forward the collapsed qubit (together with the remaining untouched ones) and observe if $\S$ accepts the session. If not, then the attacker knows that he introduced an error and hence that the first qubit must have been encoded and measured using the $\times$-basis, which gives him one bit of information on the key $w $. The error-tolerant scheme seems to prevent this particular attack, but it is by no means clear that it is secure against {\em any} man-in-the-middle attack.

To defeat a man-in-the-middle attack that tampers with the quantum
communication, we perform a check of correctness on a random subset.
The check allows to detect if the attacker tampers too much with the
quantum communication, and the scheme can be aborted before sensitive
information is leaked to the attacker.  In order to protect the
classical communication, one might use a standard
information-theoretic authentication code. However, the key for such a
code can only be securely used a limited number of times. A similar
problem occurs in QKD: even though a successful QKD execution produces
fresh key material that can be used in the next execution, the
attacker can have the parties run out of authentication keys by
repeatedly enforcing the executions to fail. In order to overcome this
problem, we will use some special
authentication scheme allowing to re-use the key under certain
circumstances, as discussed in Sect.~\ref{sec:EXTR-MAC}.

%%%%%%%%%%%%%%%%%%%%%%%%%%%%%%%%%%%%%%%%%%%%%%%%%%%%%%%%%%%%%%%%%
\subsection{The Setting}
%%%%%%%%%%%%%%%%%%%%%%%%%%%%%%%%%%%%%%%%%%%%%%%%%%%%%%%%%%%%%%%%%

Similar to before, we assume that the user $\U$ and the server $\S$
share a not necessarily uniform, low-entropy key $w $. In order to
handle the stronger security requirements of this section, we have
to assume that $\U$ and $\S$ in addition share a uniform high-entropy
key $k$. We require that a man-in-the-middle attacker can do no better
that making a guess $w'$ at $w$, and if his guess is
incorrect then he learns no more information on $w $ besides that his
guess is wrong, and essentially no information on $k$. More formally:

\begin{definition}
We say that an identification protocol is
secure against man-in-the-middle attacks by $\sf E$ with error $\eps$ if,
whenever the initial state of $\sf E$ 
is independent of the keys $W$ and $K$,
there exists $W'$, independent of $W$, such that 
the common state $\rho_{K W \regE}$ after the execution of the protocol satisfies
$$
\rho_{K W W' \regE|W' \neq W} \epsclose \rho_K \otimes \rho_{W\leftrightarrow W'\leftrightarrow  \regE| W' \neq W}.
$$
\end{definition}

Furthermore, we require security against impersonation attacks, as defined in the previous section, {\em even if the
  dishonest party knows~$k$}. It follows that $k$ can for instance
be stored on a smart card, and security is still guaranteed even
if the smart card gets stolen, assuming that the theft is noticed and
the corresponding party does/can not execute the scheme anymore. We
would also like to stress that by our security notion, not only $w $
but also $k$ may be safely reused, even if the scheme was under
attack.

%%%%%%%%%%%%%%%%%%%%%%%%%%%%%%%%%%%%%%%%%%%%%%%%%%%%%%%%%%%%%%%%%
\subsection{An Additional Tool: Extractor MACs}\label{sec:EXTR-MAC}
%%%%%%%%%%%%%%%%%%%%%%%%%%%%%%%%%%%%%%%%%%%%%%%%%%%%%%%%%%%%%%%%%

An important tool used in this section is an authentication scheme,
i.e., a Message Authentication Code (MAC), that also acts as an
extractor, meaning that if there is high min-entropy in the message,
then the key-tag pair cannot be distinguished from the key and a
random tag. Such a MAC, introduced in \cite{DKRS06}, is called an
extractor MAC, EXTR-MAC for short.  For instance
\smash{$MAC^*_{\alpha,\beta}(x) = [\alpha x]+\beta$}, where $\alpha,x
\in GF(2^n)$, $\beta \in GF(2^\ell)$ and $[\,.\,]$, denotes truncation
to the $\ell$ first bits, is an EXTR-MAC: impersonation and substitution
probability are $1/2^{\ell}$, and, for an arbitrary message $X$ and ``side information'' $Z$, a
random key $K = (A,B)$ and the corresponding tag $T = [A\cdot X]+B$,
the tuple $(T,K,Z)$ is \smash{$\big(\frac12 \cdot 2^{-\frac12(\Hmin[\eps](X|Z)-\ell)}+2\eps\big)$}-close to $(U,K,Z)$,
where $U$ is the uniform distribution, respectively, $\rho_{T K Z
  \regE}$ is \smash{$\big(\frac12 \cdot 2^{-\frac12(\Hmin[\eps](X|Z)-q-\ell)}+2\eps\big)$}-close to
${\textstyle\frac{1}{2^\ell}}\I \otimes \rho_{K Z \regE} =
{\textstyle\frac{1}{2^\ell}}\I \otimes \rho_K \otimes \rho_{Z \regE}$ if
we allow a $q$-qubit state $\regE$ that may depend only on~$X$ and $Z$. A
useful feature of an EXTR-MAC is that if an adversary gets to see the
tag of a message on which he has high min-entropy, then the key for
the MAC can be safely re-used (sequentially). Indeed, closeness of the
real state, $\rho_{T K \regE}$, to the ideal state,
${\textstyle\frac{1}{2^\ell}}\I \otimes \rho_{K \regE} =
{\textstyle\frac{1}{2^\ell}}\I \otimes \rho_K \otimes \rho_{\regE}$ ,
means that no matter how the state evolves, the real state behaves
like the ideal one (except with small probability), but of course in
the ideal state, $K$ is still ``fresh'' and can be reused.

%%%%%%%%%%%%%%%%%%%%%%%%%%%%%%%%%%%%%%%%%%%%%%%%%%%%%%%%%%%%%%%%%
\subsection{The Scheme}
%%%%%%%%%%%%%%%%%%%%%%%%%%%%%%%%%%%%%%%%%%%%%%%%%%%%%%%%%%%%%%%%%

As for \QID, let $\code:\W\rightarrow\set{+,\times}^n$ be the encoding function of a binary code of length $n$ 
with $m = |\W|$ codewords and minimal distance $d$, and for parameter $\ell$, let ${\cal F}$ and ${\cal G}$ be strongly universal-2 classes of hash functions from $\set{0,1}^n$ to $\set{0,1}^{\ell}$ and $\W$ to $\set{0,1}^{\ell}$, respectively.
Also, let $MAC^*:{\cal K} \times {\cal M} \rightarrow \set{0,1}^{\ell}$ be an EXTR-MAC with an arbitrary key space $\cal K$, a message space $\cal M$ that will become clear later, and an error probability $2^{-\ell}$. 
Furthermore, let $\set{syn_j}_{j \in \cal J}$ be the family of syndrome functions\footnote{We agree on the following convention: for a bit string $y$ of arbitrary length, $syn_j(y)$ is to be understood as $syn_j(y0\cdots0)$ with enough padded zeros if its bit length is smaller than $n'$, and as $\big(syn_j(y'),y''\big)$, where $y'$ consist of the first $n'$ and $y''$ of the remaining bits of $y$, if its bit length is bigger than $n'$. } corresponding to a family ${\cal C} = \set{C_j}_{j \in \cal J}$ of linear error correcting codes of size $n' = n/2$, as discussed in Section~\ref{sec:noise}: any $C_j$ allows to efficiently correct a $\delta$-fraction of errors for some constant $\delta > 0$, and for a random $j \in \cal J$, the syndrome of a string with $t = (\frac14-\lambda)d - \log(m) - 3\ell$ bits of min-entropy is $2^{-(t-2q)/4}$-close to uniform (given $j$ and any $q$-qubit state) for some $\lambda > 0$.

Recall, by the set-up assumption, the user $\U$ and the server $\S$
share a password $w \in \W$ as well as a uniform high-entropy key,
which we define to be a random authentication key $k \in \cal K$.  The
resulting scheme \xQID is given in the box below.

\begin{myfigure}{h}
\begin{myprotocol}{\xQID}
\item\label{it:xprepare} $\U$ picks $x \in_R \set{0,1}^n$ and $\theta \in_R \{+,\times \}^n$, and sends the $n$-qubit state $\ket{x}_\theta$ to $\S$. 
Write $I_w \assign  \Set{i}{\theta_i \!=\! \code(w)_i}$. 
\item\label{it:xmeasure} $\S$ picks a random subset $T \subset \set{1,\ldots,n}$ of size $\ell$, it computes $c = \code(w)$, replaces every $c_i$ with $i \in T$ by $c_i \in_R\set{+,\times}$ and measures $\ket{x}_\theta$ in basis $c$. Let $x'$ be the outcome, and let $test' \assign x'|_T$. 
 \item\label{xbound-applies} $\U$ sends $\theta$, $j \in_R \cal J$, $s
 \assign syn_j(x|_{I_w})$, and $f \in_R {\cal F}$ to $\S$.  
\item $\S$ picks $g \in {\cal G}$, and sends $T$ and $g$ to $\U$. 
\item $\U$ sends $test \assign x|_T$, $z \assign f(x|_{I_w}) \oplus
  g(w)$ and $tag^* \assign MAC^*_{k}(\theta,j,s,f,g,T,test,z,x|_{I_w})$ to $\S$. 
\item $\S$ recovers $x|_{I_w}$ from $x'|_{I_w}$ with the help of $test$ and $s$, and it accepts if and only if (1) $tag^*$ verifies correctly, (2) $test$ coincides with $test'$ wherever the bases coincide, and (3) $z = f(x|_{I_w}) \oplus g(w)$. 
\end{myprotocol}
\end{myfigure}

\begin{proposition}[Security against man-in-the-middle]\label{prop:MiMA}
Assume that the quantum memory of $\E$ is of size at most $q$ qubits at step 3 of \xQID.
Then
\xQID is secure against man-in-the-middle attacks by $\sf E$ with error $\eps$, where 
 $$\eps = \negl\bigl(({\textstyle\frac14}-\lambda)d - \log(m) - 2 q - 3 \ell\bigr) + \negl\bigl(\sigma(\lambda) d - \log(m)\bigr) + \negl(\ell)
 $$ 
 for an arbitrary $0 < \lambda < \frac14$. 
\end{proposition}

\begin{proof}
We use capital letters ($W$, $\Theta$, etc.) for the values ($w $,
$\theta$, etc.) occurring in the scheme whenever we view them as
random variables, and we write $X_W$ and $X'_W$ for the random
variables taking values $x|_{I_w}$ and $x'|_{I_w}$, respectively. To simplify the argument, we neglect error probabilities that are 
of order $\eps$, 
as well as linear fractions that can be chosen arbitrarily small. We merely give indication of a small error by (sometimes) using the word ``essentially''.  

First note that due to the security of the MAC and its key, if the attacker substitutes $\theta,j,s,f,g,T,test$ or $z$, or if $\S$ recovers an incorrect string as $x|_{I_w}$, then $\S$ will reject at the end of the protocol. 
We can define $W'$ (independent of $W$) as in the proof of Proposition~\ref{prop:Usec} such that if $W \neq W'$ then $X_W$ has essentially $d/4 - \log(m)$ bits of smooth min-entropy, given $W,W'$ and~$\Theta$. 
Furthermore, given $TAG^*,F(X_W),TEST$ (as well as
$K,F,T,W,W'$ and $\Theta$), $X_W$ has still essentially $t = d/4 -
\log(m) - 3\ell$ bits of smooth min-entropy, if $W \neq W'$. By the property
of the code family $\cal C$, it follows that if $t > 2q$ with a linear gap then the syndrome $S = syn_J(X_W)$ is essentially random and independent of $J,TAG^*,F(X_W),TEST,K,F,T,W,W',\Theta$ and $\regE$, conditioned on $W \neq W'$. Furthermore, it follows from the privacy-amplifying property
of $MAC^*$ and of $f$ that if $d/4 - \log(m) - 2\ell > q$ with a linear gap, then the set of values
$(TAG^*,F(X_W))$ is essentially random and independent of $K,F,TEST,T,W,W',\Theta$ and $\regE$, conditioned on $W \neq W'$. 
Finally, $K$ is independent of the rest, and $\regE$ is independent of $K,F,TEST,T,W,\Theta$. 
It follows that $\rho_{K W W' \regE|W' \neq W} \approx \rho_K \otimes \rho_{W\leftrightarrow W'\leftrightarrow  \regE| W' \neq W}$, {\em before} he learns $\S$'s decision to accept or reject. 

It remains to argue that $\S$'s decision does not give any additional information on $W$. 
We will make a case distinction, which does not depend on $w$, and we
will show for both cases that $\S$'s decision to accept or reject is
independent of $w$, which proves the claim.  But first, we need
the following observation.  Recall that outside of the test set $T$,
$\S$ measured in the bases dictated by $w$, but within $T$ in random
bases.  Let $I'_w$ be the subset of positions $i \in I_w$ with $c_i =
\code(w)_i$ (and thus also $=\theta_i$), and let $T' = T \cap I'_w$.
In other words, we remove the positions where $\S$ measured in the
``wrong'' basis. The size of $T'$ is essentially $\ell/4$, and given
its size, it is a random subset of $I'_w$ of size $|T'|$. It follows
from the theory of random sampling
that $\nu\bigr(x|_{I'_w},x'|_{I'_w}\bigr)$
essentially equals $\nu\bigl(x|_{T'},x'|_{T'}\bigr)$ (except with
probability negligible in the size of $T'$), where $\nu(\cdot,\cdot)$
denotes the fraction of errors between the two input strings.  
Furthermore,
since the set $V = \Set{i\in T}{\theta_i = c_i}$ of positions where $\U$
and $\S$ compare $x$ and $x'$ is a superset of $T'$ of essentially
twice the size, $\nu\bigl(x|_{V},x'|_{V}\bigr)$ is essentially lower
bounded by $\frac12\,\nu\bigl(x|_{T'},x'|_{T'}\bigr)$.  Putting things
together, we get that $\nu\bigr(x|_{I'_w},x'|_{I'_w}\bigr)$ is
essentially upper bounded by $2\,\nu\bigl(x|_{V},x'|_{V}\bigr)$. Also
note that $\nu\bigl(x|_{V},x'|_{V}\bigr)$ does not depend on $w$. We
can now do the case distinction:
{\it Case~1: } If $\nu\bigl(x|_{V},x'|_{V}\bigr) \leq \frac{\delta}{2}$ (minus an arbitrarily small value), then $x|_{I'_w}$ and $x'|_{I'_w}$ differ in at most a $\delta$-fraction of their positions, and thus $\S$ correctly recovers 
$x|_{I_w}$ (using $test = x|_T$ to get $x|_{I_w\setminus I'_w}$ and using $s$ to correct the rest), no matter what $w$ is, and it follows that $\S$'s decision only depends on the attacker's behavior, but not on $w$. 
{\it Case~2: } Otherwise, $\S$ is guaranteed to get the correct $test = x|_T$ (or else rejects) 
and thus rejects as $test$ and $test'$, restricted to $V$, differ 
in more than a $\frac{\delta}{2}$-fraction of their positions. Hence,
$\S$ always rejects in case~2.
\qed
\end{proof}
For a dishonest user or server who knows $k$ (but not $w$), breaking \xQID is equivalent to breaking \QID, up to a change in the parameters. Doing the maths on the parameters similarly to the proof of Theorem~\ref{thm:impersonation} (namely, choosing $\ell = \frac14\big((\frac14-\lambda)d+\log(m)-2q\big)$ whence $\eps = \negl\bigl((\frac14-\lambda)d-7\log(m)-2q\bigr)$),
it then follows:

\begin{theorem}
If $\Hmin(W) \geq 1$, then the identification scheme \xQID is $\eps$-secure against a man-in-the-middle attacker with quantum memory bound $q$, and, even with a leaked $k$, \xQID is $\eps$-secure against impersonation attacks for any unbounded user and for any server with quantum memory bound $q$, where 
$$\eps 
= \negl\bigl(({\textstyle\frac14}\!-\!\lambda) \mu n - 7\log(m) - 2q\bigr) + \negl\bigl(\sigma(\lambda) \mu n - \log(m)\bigr)
$$ 
for $\mu = h^{-1}(1-\log(m)/n)$ and an arbitrary $0 < \lambda < \frac14$.  
In particular, if $\log(m)$ is sublinear in $n$, $\eps$ is negligible in $n- 16q$.
\end{theorem}

%%%%%%%%%%%%%%%%%%%%%%%%%%%%%%%%%%%%%%%%%%%%%%%%%%%%%%%%%%%%%%%%%
%%%%%%%%%%%%%%%%%%%%%%%%%%%%%%%%%%%%%%%%%%%%%%%%%%%%%%%%%%%%%%%%%

It is easy to see that \xQID can tolerate a noisy quantum communication up to any error rate 
$\phi < \delta$. Similar to the discussion in Section~\ref{sec:noise}, tolerating a 
higher error rate requires the bound on the adversary's quantum memory to be smaller
but still linear in the number of qubits transmitted. Imperfect sources can also
be addressed in a similar way as for \QID. It follows that \xQID\ can also be shown secure in
the $(\phi,\eta)$-weak quantum model provided $\phi$ and $\eta$ are
small enough constants.
%
%%%
%%
%%
%%
%%%

%%%%%%%%%%%%%%%%%%%%%%%%%%%%%%%%%%%%%%%%%%%%%%%%%%%%%%%%%%%%%%%%%
\section{Application to QKD}
%%%%%%%%%%%%%%%%%%%%%%%%%%%%%%%%%%%%%%%%%%%%%%%%%%%%%%%%%%%%%%%%%

As already pointed out in Section~\ref{sec:xapproach}, current QKD
schemes have the shortcoming that if there is no classical channel
available that is authenticated by physical means, and thus messages
need to be authenticated by an information-theoretic authentication
scheme, an attacker can force the parties to run out of
authentication keys simply by making an execution (or several
executions if the parties share more key material) fail. Even worse,
even if there is no attacker, but some execution(s) of the QKD scheme
fails due to a technical problem, parties could still run out of
authentication keys because it may not be possible to distinguish between
an active attack and a technical failure. This shortcoming could make the technology
impractical in situations where denial of service attacks or
technical interruptions often occur. 

The identification scheme \xQID from the previous section immediately gives a QKD scheme {\em in the bounded-quantum-storage model} that allows to re-use the authentications key(s). Actually, we can inherit the key-setting from \xQID, where there are two keys, a human-memorizable password and a uniform, high-entropy key, where security is still guaranteed even if the latter gets stolen and the theft is noticed. 
In order to agree on a secret key $sk$, the two parties execute \xQID, and extract $sk$ from $x|_{I_w}$ by applying yet another strongly universal-2 function, for instance chosen by $\U$ in step~\ref{xbound-applies} and authenticated together with the other information in Step 5. Here, $n$ needs to be increased accordingly to have the additional necessary amount of entropy in $x|_{I_w}$. The analysis of \xQID immediately implies that if honest $\S$ accepts, then he is convinced that he shares $sk$ with the legitimate $\U$ which knows $w$. In order to convince $\U$, $\S$ can then use part of $sk$ to one-time-pad encrypt $w$, and send it to $\U$. The rest of $sk$ is then a secure secret key, shared between $\U$ and $\S$. In order to have a better ``key rate'', instead of using $sk$ (minus the part used for the one-time-pad encryption) as secret key, one can also run a standard QKD scheme on top of \xQID and use $sk$ as a one-time authentication key.

%%%%%%%%%%%%%%%%%%%%%%%%%%%%%%%%%%%%%%%%%%%%%%%%%%%%%%%%%%%%%%%%%
%             bibliography                                      %
%%%%%%%%%%%%%%%%%%%%%%%%%%%%%%%%%%%%%%%%%%%%%%%%%%%%%%%%%%%%%%%%%

\bibliographystyle{alpha}

\bibliography{crypto,qip,procs}

%%%%%%%%%%%%%%%%%%%%%%%%%%%%%%%%%%%%%%%%%%%%%%%%%%%%%%%%%%%%%%%%%
%               appendix                                        %
%%%%%%%%%%%%%%%%%%%%%%%%%%%%%%%%%%%%%%%%%%%%%%%%%%%%%%%%%%%%%%%%%

\begin{appendix}

\section{Proofs}

\subsection{Proof of Lemma~\ref{lemma:decomp}} \label{app:decomp}

Writing $p = P[\ev]$ and $\bar{p} = P[\bar{\ev}]$ we indeed get
\begin{align*}
\rho_{X\leftrightarrow Y \leftrightarrow \regE} 
&= \sum_{x,y} P_{XY}(x,y) \proj{x} \otimes \proj{y} \otimes \rho^{y}_{\regE} \\
&= \sum_{x,y} \big(p \cdot P_{XY|\ev}(x,y) +  \bar{p} \cdot P_{XY|\bar{\ev}}(x,y)\big) \proj{x} \otimes \proj{y} \otimes \big(p\cdot \rho^{y}_{\regE|\ev} + \bar{p} \cdot \rho^{y}_{\regE|\bar{\ev}}\big) \\
&= p^2\cdot\sum_{x,y} P_{XY|\ev}(x,y) \proj{x} \otimes \proj{y} \otimes \rho^{y}_{\regE|\ev}  + (1-p^2)\cdot\tau \\
&= p^2\cdot\rho_{X\leftrightarrow Y \leftrightarrow \regE|\ev} + (1-p^2)\cdot\tau
\end{align*}
for some density matrix $\tau$.
If $\ev$ is independent of $X$ and $Y$, so that $P_{XY} = P_{XY|\ev} = P_{XY|\bar{\ev}}$, then 
\begin{align*}
\rho_{X\leftrightarrow Y \leftrightarrow \regE} 
&= \sum_{x,y} P_{XY}(x,y) \proj{x} \otimes \proj{y} \otimes \rho^{y}_{\regE} \\
&= \sum_{x,y} P_{XY}(x,y) \proj{x} \otimes \proj{y} \otimes \big(p\cdot \rho^{y}_{\regE|\ev} + \bar{p} \cdot \rho^{y}_{\regE|\bar{\ev}}\big) \\
&= p\cdot\sum_{x,y} P_{XY|\ev}(x,y) \proj{x} \otimes \proj{y} \otimes \rho^{y}_{\regE|\ev}  + \bar{p}\cdot\sum_{x,y} P_{XY|\bar{\ev}}(x,y) \proj{x} \otimes \proj{y} \otimes \rho^{y}_{\regE|\bar{\ev}} \\
&= p\cdot\rho_{X\leftrightarrow Y \leftrightarrow \regE|\ev} + \bar{p}\cdot\rho_{X\leftrightarrow Y \leftrightarrow \regE|\ev} \,.
\end{align*}
\qed

\subsection{Proof of Lemma~\ref{lemma:ES}}\label{app:ES}

  For any pair $i \neq j$ let $\ev_{ij}$ be an event such that
  $P[\ev_{ij}] \geq 1 - \eps$ and 
  \begin{equation} \label{eq:assumption} \sum_z P_Z(z) \cdot
    \max_{x_i,x_j} P_{X_iX_j\ev_{ij}|Z}(x_i,x_j|z) \leq 2^{-\alpha}
  \end{equation}
  for all $x_i \in \mX_i$, $x_j \in \mX_j$ and $z \in \cal Z$. By
  assumption, such events exist.\footnote{In case $\eps = 0$, i.e.,
    $\alpha$ lower bounds the ordinary (rather then the smooth)
    min-entropy, the $\ev_{ij}$ are the events ``that always occur''
    and can be ignored from the rest of the analysis. } For any
  $j=1,\ldots,m-1$ define
$$
L_j =
\Set{(x_1,\ldots,x_m,z)}{P_{X_1|Z}(x_1|z),\ldots,P_{X_{j-1}|Z}(x_{j-1}|z)
  < 2^{-\alpha/2} \wedge P_{X_j|Z}(x_j|z) \geq 2^{-\alpha/2}}
$$
Informally, $L_j$ consists of the tuples $(x_1,\ldots,x_m,z)$, where
$x_j$ has ``large'' probability given $z$ whereas all previous entries
have small probabilities. We define $V$ as follows. We let $V$ be the
index $j \in \set{1,\ldots,m-1}$ such that $(X_1,\ldots,X_m,Z) \in
L_j$, and in case there is no such $j$ we let $V$ be $m$. Note that if
there does exist such an $j$ then it is unique.

We need to show that this $V$ satisfies the claim. Fix $j \in
\set{1,\ldots,m}$. Clearly, for $i < j$,
\begin{align} \begin{split} \label{eq:ESL1}
  \sum_z P_Z(z) \cdot \max_{x_i} P_{X_i V \ev_{ij} |Z}(x_i,j|z) &\leq
  \sum_z P_Z(z) \cdot \max_{x_i} P_{X_i V |Z}(x_i,j|z)\\
&=  \sum_z P_Z(z) \cdot \max_{x_i} P_{X_i|Z}(x_i|z) P_{V|X_i
    Z}(j|x_i,z) < 2^{-\alpha/2} \, .
\end{split} \end{align}
Indeed, either $P_{X_i|Z}(x_i|z) < 2^{-\alpha/2}$ or $P_{V|X_i
  Z}(j|x_i,z) = 0$ by definition of $V$. 
Consider now $i > j$. Note that  
\begin{align} \begin{split} \label{eq:ESL2}
  \sum_z P_Z(z) \cdot \max_{x_i} P_{X_iV\ev_{ij}|Z}(x_i,j|z) &= \sum_z
  P_Z(z) \cdot \max_{x_i} \sum_{x_j} P_{X_i X_j
    V\ev_{ij}|Z}(x_i,x_j,j|z)\\
  &\leq 2^{\alpha/2} \sum_z P_Z(z) \cdot \max_{x_i,x_j} P_{X_i
    X_j\ev_{ij}|Z}(x_i,x_j|z) \leq 2^{-\alpha/2} \, ,
\end{split} \end{align}
where the last inequality follows from the
assumption~\eqref{eq:assumption} and the first is a consequence of the
fact that the number of non-zero summands (in the sum over $x_j$)
cannot be larger than $2^{\alpha/2}$, because for any $x_j$ with
$P_{X_i X_j V \ev_{ij} | Z}(x_i,x_j,j|z) > 0$, it also holds that
$P_{X_j|Z}(x_j|z) \geq 2^{-\alpha/2}$ and the sum over all those $x_j$
would exceed 1 if there were more than $2^{\alpha/2}$ summands.
Note that per-se, $\ev_{ij}$ is only defined in the probability space
given by $X_i$, $X_j$ and $Z$, but it can be naturally extended to the
probability space given by $X_1,\ldots,X_n,Z,V$ by assuming it to be
independent of anything else when given $X_i,X_j,Z$, so that e.g.\
$P_{X_i V \ev_{ij}|Z}$ is indeed well-defined.

Consider now an independent random variable $W$ with $\Hmin(W) \geq 1$. 
By the assumptions on $W$ it holds that $P[V\!\neq\!W] \geq \frac12$
and $P_{X_W V W Z}(x_i,j,i,z) = P_{X_i V W Z}(x_i,j,i,z) = P_{X_i V
  Z}(x_i,j,z) P_W(i)$.  In the probability space determined by the
random variables $X_1,\ldots,X_n,V,W,Z$ and all of the events $\ev_{ij}$, define the event
$\ev$ as $\ev \assign \ev_{WV}$, so that $P_{X_W V W \ev|Z}(x_i,j,i|z) =
P_{X_i V W \ev_{ij}|Z}(x_i,j,i|z) = P_{X_i V \ev_{ij}|Z}(x_i,j|z)
P_W(i)$. Note that
$$
P[\bar{\ev}]
= \sum_{i,j} P_{VW \bar{\ev}_{WV}}(j,i)
= \sum_{i,j} P_{V \bar{\ev}_{ij}}(j) P_W(i)
\leq \sum_{i,j} P[\bar{\ev}_{ij}] P_W(i)
\leq m \eps
$$
and thus $P[\bar{\ev}|V\!\neq\!W] \leq P[\bar{\ev}]/P[V\!\neq\!W]
\leq 2m \eps$.  From the above, it follows that
\begin{align*}
  \guess&(X_W,\ev|VWZ,V\neq W) = \sum_{z,i,j} \max_x P_{X_W V W Z \ev|V \neq
    W}(x,j,i,z)
  \leq 2 \sum_{z,i \neq j} \max_x P_{X_W V W Z \ev}(x,j,i,z)\\
  &=2 \sum_{z,i \neq j} P_Z(z) \cdot \max_x P_{X_W V W \ev|Z}(x,j,i|z)
  = 2 \sum_{z,i \neq j} P_Z(z) \cdot \max_{x_i} P_{X_i V \ev_{ij}|Z}(x_i,j|z) \cdot P_W(i) \\
  &= 2 \sum_{i} P_W(i) \sum_{j\neq i} \sum_z P_Z(z) \cdot
  \max_{x_i} P_{X_i V \ev_{ij}|Z}(x_i,j|z) \leq 2m \cdot 2^{-\alpha/2} \, ,
\end{align*}
where we used \eqref{eq:ESL1} and \eqref{eq:ESL2} in the last
inequality.
The claim now follows by definition of $\Hmin$.
\qed

\end{appendix}

\end{document}